\documentclass[letterpaper,10pt, conference]{IEEEtran}
\usepackage[T1]{fontenc}
\usepackage[latin1]{inputenc}
\usepackage{float}
\usepackage{amsmath}
\usepackage[noadjust]{cite}
\usepackage{amssymb}
\usepackage{enumerate}
\usepackage{amsthm}
\usepackage{graphics,epsfig,psfrag}
\usepackage{subfigure,epsf,pstricks,subfloat}
\usepackage{bm}
\usepackage{url}
\floatstyle{ruled}
\newfloat{algorithm}{tbp}{loa}
\floatname{algorithm}{Routine}

\usepackage{algorithmic}
\algsetup{linenodelimiter= }

\renewcommand{\algorithmicreturn}{\textbf{Return:}}
\def\bfw{{\mathbf w}}
\def\bfx{{\mathbf x}}
\def\bfy{{\mathbf y}}
\def\bfz{{\mathbf z}}

\def\w{{\bf w}}

\title{Sparse Recovery with Graph Constraints: Fundamental Limits and Measurement Construction}

\author{Meng Wang \ \ \ \ \ Weiyu Xu \ \ \ \ \ Enrique Mallada \ \ \ \ \   Ao Tang
\\
School of ECE, Cornell University, Ithaca, NY 14853, USA }

\usepackage{graphicx}
\usepackage{epstopdf}
\newtheorem{theorem}{Theorem}

\newtheorem{defi}{Definition}
\newtheorem{prop}{Proposition}

\begin{document}

\maketitle \thispagestyle{empty} \pagestyle{empty}

\begin{abstract}
This paper addresses the problem of sparse recovery with graph constraints in the sense that we can take additive measurements over nodes only if they induce a connected subgraph. 
We provide explicit measurement constructions for several special graphs. 
A general measurement construction algorithm is also proposed and evaluated. For any given graph $G$ with $n$ nodes, we derive order optimal upper bounds of the minimum number of measurements needed to recover any $k$-sparse vector over $G$ ($M^G_{k,n}$). 
Our study suggests that $M^G_{k,n}$ may serve as a graph connectivity metric.
\end{abstract}

\section{Introduction} \label{sec:intro}

Network monitoring is an important module in the operation and management of communication networks, where network performance
characteristics, 
such as traffic transmission rates and router queueing delays, should be monitored. 
Since monitoring each object in the network directly can be 
operationally difficult or even infeasible, 
the topic of inferring internal characteristics using information from indirect end-to-end (aggregate) measurements, known as Network Tomography, has been widely explored recently \cite{BDPT02,CHNY02,CBSK07,Duffield06,GR11,NT06,ZCB06}. 

%

In practice, the total number of aggregate measurements we can take is small compared with the size of the network. 
However, we can indeed extract the most dominating elements of a high-dimensional signal from low-dimensional non-adaptive measurements. 
With the signal itself being sparse, i.e. most entries are zero,  the recovered signal can be exact even though the number of measurements is much smaller than the dimension of the signal. One practical example is that only a small number of bottleneck links in the communication networks experience large delays.
\textit{Sparse Recovery} addresses the problem of recovering sparse high-dimensional signals from low-dimensional measurements, and has two different  but closely related problem formulations.
One is \textit{Compressed Sensing} \cite{CaT05,CaT06,DoT05,Don06,BGIKS08,FR11}, where the signal is represented by a high-dimensional real vector, 
and an aggregate measurement is the arithmetical sum of the corresponding real entries. The other is \textit{Group Testing} \cite{Dorfman43, DH00},
where the high-dimensional vector is logical, and a measurement is a logical disjunction (\textbf{OR}) on the corresponding logical values.

One key question in both compressed sensing and group testing is to design a small number of non-adaptive measurements (either real or logical) such that all the vectors (either real or logical) up to certain sparsity (the support size of a vector) can be correctly recovered. Most existing results, however, rely critically on the assumption that any subset of the values can be aggregated together \cite{CaT05,Don06}, which is not realistic in the network monitoring problem. Here 
 only objects that can form a path or a cycle on the graph \cite{GR11}, or induce a connected subgraph can be combined together in the same measurement.  Only a few recent works consider graph topological constraints in compressed sensing \cite{CPR07,FR11,HBRN08, XMT11,ZRWQ09} and group testing \cite{BTH11,CKMS10,HPWYC07,TWHR11,WHTJ10}.

Though motivated by the network monitoring application, 
 beyond networks. Indeed, this formulation abstractly models that certain elements cannot be measured together in a complex system. Thus, our work can be useful to other applications besides network tomography.

Here are the main contributions of this paper.

\noindent{\bf (1)} We provide explicit measurement constructions for different graphs. Moveover, the number of our measurements improves over the existing estimates (e.g. \cite{CKMS10,XMT11}) of the minimum number of measurements required to recover sparse vectors over graphs. (Section \ref{sec:special})

\noindent{\bf (2)} 
We propose a design guideline based on \textit{$r$-partition} for general graphs and further show some of its properties. (Section \ref{sec:bound})

\noindent{\bf (3)} A simple measurement design algorithm is proposed for general graphs. (Section \ref{sec:algo}) We evaluate its performance both theoretically and numerically. (Section \ref{sec:simu})

We now start with Section \ref{sec:model} to introduce the model and problem formulation.

\section{Model and Problem Formulation}\label{sec:model}

Consider a graph $G=(V,E)$, where $V$ denotes the set of nodes with cardinality $|V|=n$ and $E$ denotes the set of links. Each node $i$ is associated with a real number $x_i$, and we say vector $\bfx=(x_i, i=1,...,n)$ is associated with $G$. Let $T=\{i~|~ x_i \neq 0\}$ denote the support of $\bfx$, and let $\|\bfx\|_0=|T|$ denote the number of non-zero entries of $\bfx$, we say $\bfx$ is a $k$-sparse vector if $\|\bfx\|_0=k$.

Let $S \subseteq V$ denote a subset of nodes in $G$. Let $E_S$ 
denote the subset of links with both ends in
$S$, then $G_S=(S, E_S)$ is the induced subgraph of $G$. We have the following two assumptions throughout the paper:\\
\noindent {\bf (A1)}: A set $S$ of nodes can be measured together in one measurement if and only if  $G_S$ is connected.

\noindent{\bf (A2)}:
 The measurement is an additive sum of values at the corresponding nodes.


 (A1) captures the graph constraints. One practical example is a sensor network where the nodes represent sensors and the links represent feasible communication between sensors. For the set $S$ of nodes that induce a connected subgraph, one node $u$ in $S$ monitors the total values corresponding to nodes in $S$. Every node in $S$ obtains values from its children, if any, on the spanning tree rooted at $u$, aggregates them with its own value and sends the sum to its parent. 
 Then the fusion center can obtain the sum of values corresponding to all the nodes in $S$ by only communicating with $u$.
 (A2) follows from the additive property of many network characteristics, e.g. delays and packet loss rates \cite{GR11}. 
 However, compressed sensing can also be applied to cases where (A2) does not hold, e.g., the measurements can be nonlinear as in \cite{WWGZMM11,Blumensath10}.

 Let $\bfy \in \mathcal{R}^m$ $(m\ll n)$ denote the vector of $m$ measurements.  Let $A$ be an $m \times n$ measurement matrix with $A_{ij}=1$ ($i=1,...,m$, $j=1,...,n$) if and only if node $j$ is included in the $i$th measurement and $A_{ij}=0$ otherwise. Then we have $\bfy=A\bfx$. We say $A$ can identify all $k$-sparse vectors if and only if $A\bfx_1 \neq A\bfx_2$ for every two different vectors $\bfx_1$ and $\bfx_2$ that are at most $k$-sparse.
 The advantage of sparse recovery is that with the non-adaptive measurement matrix $A$, it can identify $n$-dimensional vectors from $m$ ($m \ll n$) measurements as long as the vectors are sparse.
\begin{figure}[h]
\begin{center}
  \includegraphics[scale=0.25]{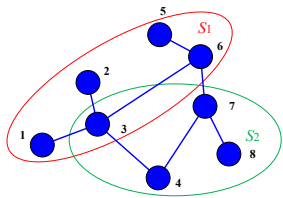}
  \caption{Network Example}\label{fig:example}
  \end{center}
\end{figure}

With the above assumptions, $A$ is a $0$-$1$ matrix and for each row of $A$, the set of nodes that correspond to `1' should form a connected induced subgraph of $G$. In Fig. \ref{fig:example}, we can measure nodes in $S_1$ and $S_2$ separately, and the 
measurement matrix is
\begin{equation}\nonumber
A=\left[ \begin{array}{cccccccc}
1 &  1 &1 &0 & 1 & 1 &0 & 0\\
0 & 0& 1 & 1 & 0 & 0& 1 & 1
\end{array}
 \right].
\end{equation}
We remark here that in group testing with graph constraints, the requirements for the measurement matrix $A$ are the same, while group testing differs from compressed sensing only in that (1) $\bfx$ is a logical vector, and (2) the operation used in each group testing measurement is the logical ``OR''. 
All arguments and results in this paper are in the compressed sensing setup if not otherwise specified, and we also compare our results with group testing for special networks. 
Note that for recovering $1$-sparse vectors, 
the numbers of measurements required by compressed sensing and group testing are the same.
%

Given a graph $G$ with $n$ nodes, let $M^G_{k,n}$ denote the minimum number of non-adaptive measurements needed  to identify all $k$-sparse vectors associated with $G$. Let $M^C_{k,n}$ denote the minimum number of non-adaptive measurements needed in a complete graph with $n$ nodes. In complete graphs, since any subset of nodes can be measured together, 
any $0$-$1$ matrix is a feasible measurement matrix. 
Existing results \cite{CaT06,BGIKS08,XH07} show that with overwhelming probability a random $0$-$1$ matrix with $O( k\log (n/k))$ rows\footnote{We use the notations  $g(n)\in O(h(n))$, $g(n) \in \Omega(h(n))$, or $g(n)=\Theta(h(n))$ if as $n$ goes to infinity, $g(n) \leq c h(n)$, $g(n) \geq c h(n)$ or $c_1 h(n) \leq g(n) \leq c_2 h(n)$ eventually holds for some positive constants $c$, $c_1$ and $c_2$ respectively.}
can identify all $k$-sparse vectors, and we can recover the sparse vector by $\ell_1$-minimization, which returns the vector with the least $\ell_1$-norm\footnote{The $\ell_p$-norm ($p\geq 1$) of $\bfx$ is $\|\bfx\|_p=(\sum_i |x_i|^p)^{1/p}$, and $\|\bfx\|_{\infty}=\max_i |x_i|$.} among those that can produce the obtained measurements. Then we have
\begin{equation}\label{eqn:MC}
M^C_{k,n} = O( k\log (n/k)).
 \end{equation}

 We will use (\ref{eqn:MC}) for the analysis of construction methods.
 Explicit constructions of measurement matrices for complete graphs also exist, e.g., \cite{AHSC09,BGIKS08,CM06,DeVore07,XH07}.  
 We will use $f(k,n)$ to denote the number of measurements to recover $k$-sparse vectors associated with the complete graph of $n$ nodes by a particular measurement construction method, and $f(k,n)$ varies for different construction methods. The key notations are summarized in Table \ref{tbl:notation}.

 \begin{table}
\caption{summary of key notations} \label{tbl:notation}
\begin{tabular}{|c|p{7.1cm}|} %
 \hline
Notation & Meaning\\
\hline
$G_S$ & Subgraph of $G$ induced by $S$\\
\hline
$M^G_{k,n}$ & Minimum number of measurements needed to recover  $k$-sparse vectors associated with $G$ of $n$ nodes.\\
\hline
$M^C_{k,n}$ & Minimum number of  measurements needed to recover  $k$-sparse vectors associated with a complete graph of $n$ nodes.\\
\hline
$f(k,n)$ & Number of measurements constructed to recover $k$-sparse vectors associated with a complete graph of $n$ nodes\\
\hline
\end{tabular}
\end{table}

The questions we would like to address in the paper are:
\begin{itemize}
\item Given graph $G$, what is the corresponding $M_{k,n}^G$?
\item How to explicitly design measurements such that the total number of measurements is close to $M_{k,n}^G$?
\end{itemize}



\section{Sparse Recovery over Special Graphs}\label{sec:special}

In this section, we consider four kinds of special graphs: one-dimensional line/ring network, ring with each node connecting to four closest neighbors, two-dimensional grid and a tree.
 We construct measurements for each graph and later generalize the construction ideas obtained here to general graphs in Section \ref{sec:general}.

\subsection{Line and Ring}\label{sec:line}

First consider one-dimensional line/ring network as shown in Fig. \ref{fig:line}. When later comparing the results here with those in Section \ref{sec:ring4} one can see that the number of measurements required to recover sparse vectors can be significantly different in two graphs that only differ from each other with a small number of links.

In a line/ring network, there is not much freedom in the measurement design since only consecutive nodes can be measured together from assumption (A1). In fact,  \cite{HPWYC07,TWHR11} show that $\lceil\frac{n+1}{2}\rceil$ (or $\lceil\frac{n}{2}\rceil$) measurements
are both necessary and sufficient to recover  $1$-sparse vectors associated with a line (or ring) network with $n$ nodes. Therefore, $\Theta(n)$ measurements are required to recover even one non-zero element associated with a line/ring network.
\begin{figure}
\begin{center}
    \includegraphics[scale=0.3]{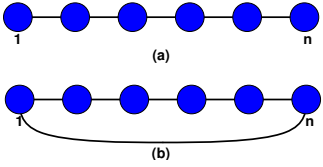}
    \caption{(a) line network (b) ring network}\label{fig:line}
\end{center}
\end{figure}

We next construct $k\lceil \frac{n}{k+1}\rceil+1$ measurements to recover $k$-sparse vectors ($k\geq2$) associated with the line/ring network. Let $t=\lceil \frac{n}{k+1}\rceil$. For every $1\leq i\leq kt+1$, the $i$th measurement goes through all the nodes from $i$ to $\min (i+t-1, n)$. 
\begin{theorem}\label{thm:linekcs}
$k\lceil \frac{n}{k+1}\rceil+1$ above measurements are
sufficient to identify all $k$-sparse vectors associated with a line/ring network with $n$ nodes. 
\end{theorem}

\begin{proof}
Consider matrix $A^{(tk+1)\times (tk+t)}$ with its $i$th row having `1's from  entry $i$ to entry $i+t-1$ and `0's elsewhere for all $1 \leq i \leq tk+1$. Then the first $n$ columns of $A$ correspond to our measurement matrix. To prove the statement, we only need to show that $A$ can identify all $k$-sparse vectors in $\mathcal{R}^{tk+t}$, which happens if and only if
every non-zero vector $\bfz$ such that $A\bfz=\bm 0$ holds has at least $2k+1$ non-zero
elements \cite{CaT05}.


For each index $1\leq k' \leq k$, define a submatrix $A_{k'}$,
which consists of the first $tk'+1$ rows and the first $tk'+t$
columns of $A$.  We claim that every non-zero vector $\bfw$ such that $A_{k'}\bfw =\bm 0$ holds has at least $2k'+1$ non-zero elements with at least
two non-zero elements in the last $t$ entries.
We prove this claim by induction over $k'$.

First consider $A_1$. Note that its first row has `1's from column $1$ to $t$, and its last row has `1's from column $t+1$ to $2t$.
Because any two columns of
the submatrix $A_{1}$ are linearly independent, for any  $\bfw \neq \bm 0$ such that $A_1 \bfw= \bm 0$, 
$\bfw$ must have at least 
three non-zero elements. 
Let $j$ be the index of the last non-zero element of $\bfw$. If $j \leq t$, consider the $j$th row of $A_1$ with its first `1' entry in the $j$th column. The inner product of the $j$th row and $\bfw$ is non-zero,  contradicting the assumption that $A_1\bfw=\bm 0$. 
Then $j \geq t+1$ must hold. Then since the inner product between $\bfw$ and the last row of $A_1$ is zero, at least two non-zero elements exist in the last $t$ entries of $\bfw$.

Now suppose the claim holds for $A_{k'}$, 
consider a non-zero
vector $\w$ such that  $A_{k'+1}\bfw=\bm 0$ holds. Note that the vector of the
first $tk'+t$ positions of $\w$, denoted by $\hat{\bfw}$, satisfies $A_{k'}\hat{\bfw}=\bm 0$. 
We remark that $\hat{\bfw} \neq \bm 0$. 
If $\hat{\bfw} = \bm 0$, let $j$ denote the index of the first non-zero element of $\bfw$, and we have $j \geq tk'+t+1$. 
Consider the $(j+1-t)$th row of $A_{k'+1}$ with its last `1' entry in column $j$. 
Then the inner product of this row with $\w$ is non-zero, which is a contradiction.

Since $\hat{\bfw} \neq \bm 0$, 
from the induction assumption, it 
 has at least $2k'+1$ non-zero elements with at least two non-zero elements in its last $t$ elements. 
Now consider the last $2t$ elements of $\w$ and the last $t+1$ measurements in $A_{k'+1}$. From a similar argument for the case of $A_1$, we know that $\w$ 
must have at least two non-zero elements in the last $t$ positions. So $\w$  has at least $2(k'+1)+1$ non-zero elements.

By induction over $k'$, every $\bfw \neq \bm 0$ satisfying $A\bfw=0$ has at least $2k+1$ non-zero entries. 
This completes the proof.
\end{proof}

Theorem \ref{thm:linekcs} 
implies that we can save about $\lfloor \frac{n}{k+1}\rfloor$ measurements but
still be able to recover $k$-sparse vectors in a line/ring network
via compressed sensing. However, for group testing associated with a line/ring network, one can check  that $n$ measurements are necessary to recover more than one non-zero element. The key is that every node should be the \textit{endpoint} at least twice, where the endpoints are the nodes at the beginning and the end of a measurement. The endpoints of a measurement can be a same node. If node $u$ is an endpoint for at most once, then it is always measured together with one of its neighbors, say $v$, if ever measured. 
Then when $v$ is `1', we cannot determine the value of $u$, either '1' or '0'. 
Therefore, to recover more than one non-zero element, we need at least $2n$ endpoints, and thus $n$ measurements. 

\subsection{Ring with nodes connecting to four closest neighbors}\label{sec:ring4}
 We know from Section \ref{sec:line} that
$\lceil n/2\rceil$ measurements are necessary to recover even one non-zero element associated with a ring
network. Now consider a graph with each node directly connecting to its
four closest neighbors as in Fig. \ref{fig:ring4} (a), denoted by $\mathcal{G}^4$. $\mathcal{G}^4$ is important to the study of small-world networks \cite{WS98}.
$\mathcal{G}^4$ has $n$ more links than the ring network, but we will show
that 
the number of
measurements required by compressed sensing to recover $k$-sparse vectors associated with $\mathcal{G}^4$ significantly reduces from $\Theta(n)$ to $O(k \log(n/k))$. 
\begin{figure*}[ht]
\centering
\begin{tabular}{c c c c}
\includegraphics[width=0.2\linewidth]{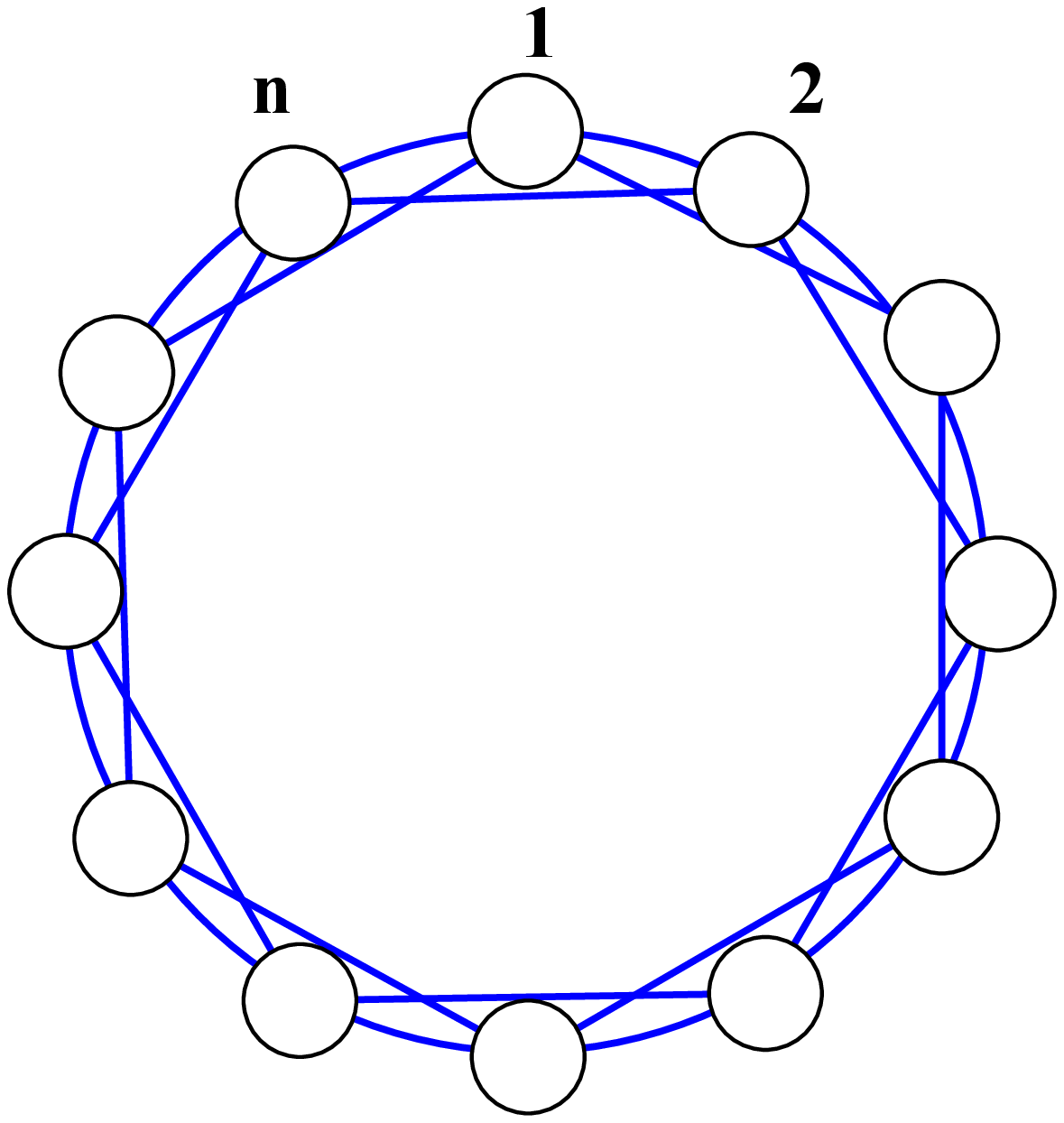}
&
\includegraphics[width=0.2\linewidth]{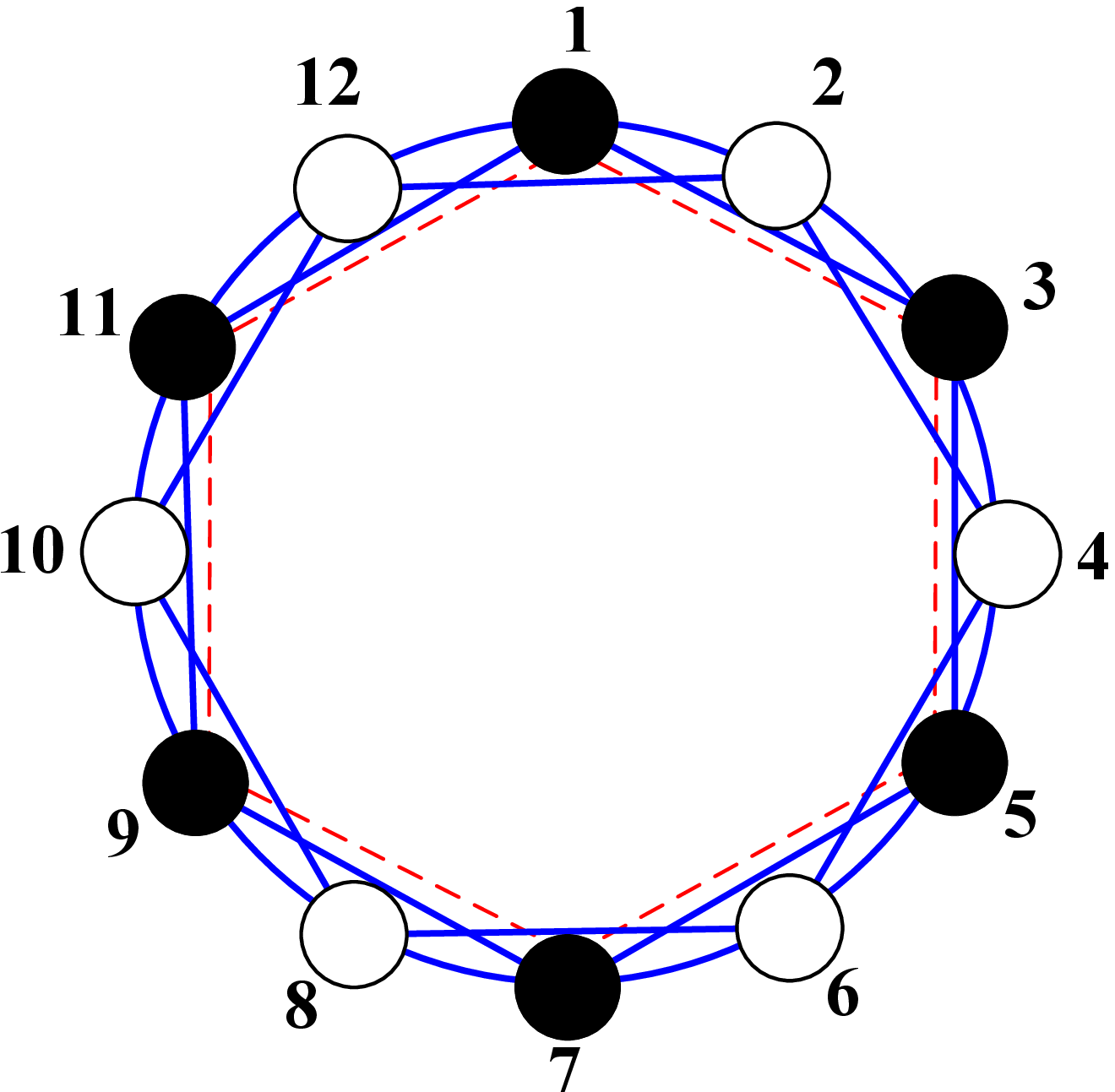}
&
\includegraphics[width=0.2\linewidth]{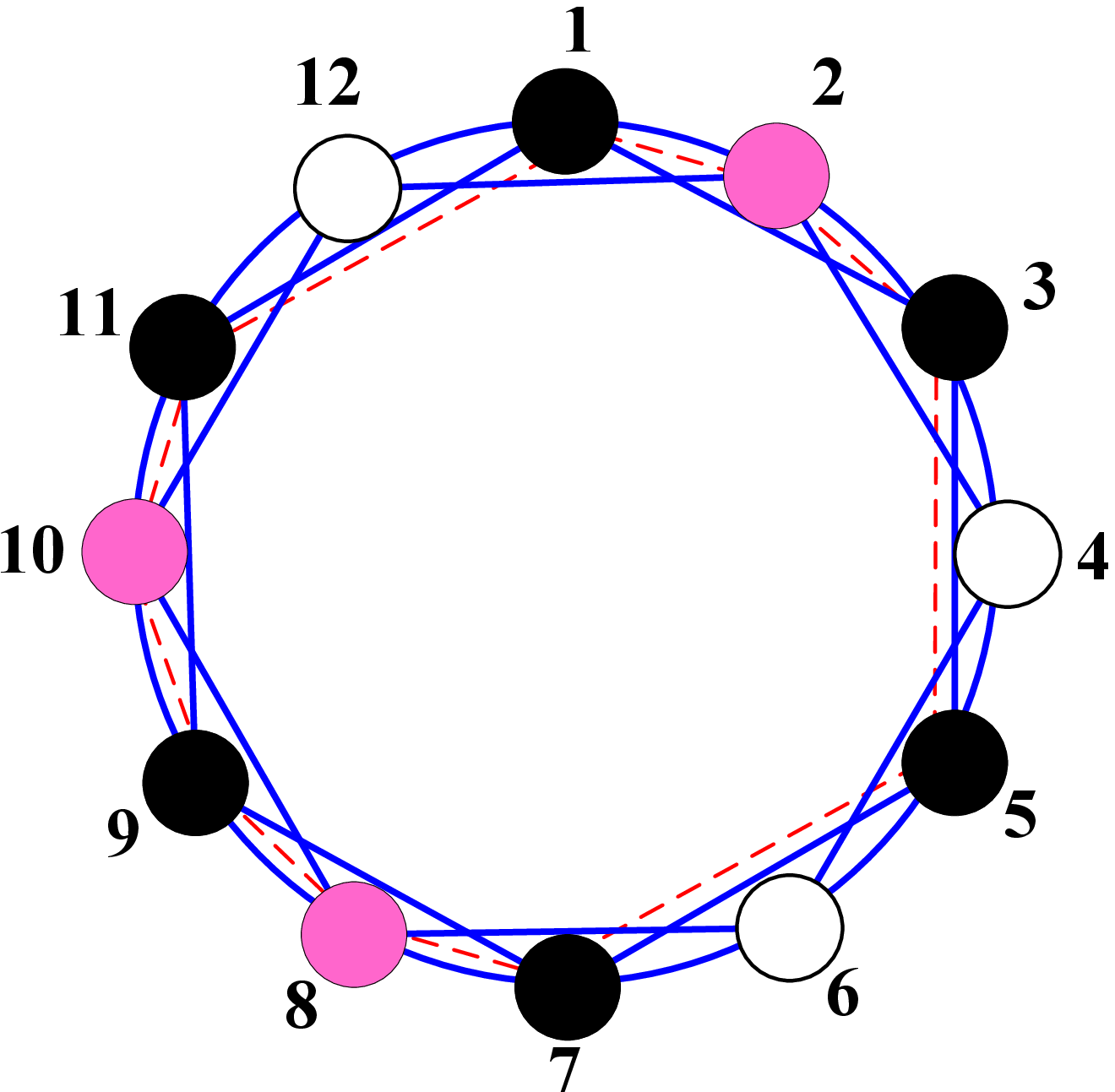}
&
\includegraphics[width=0.2\linewidth]{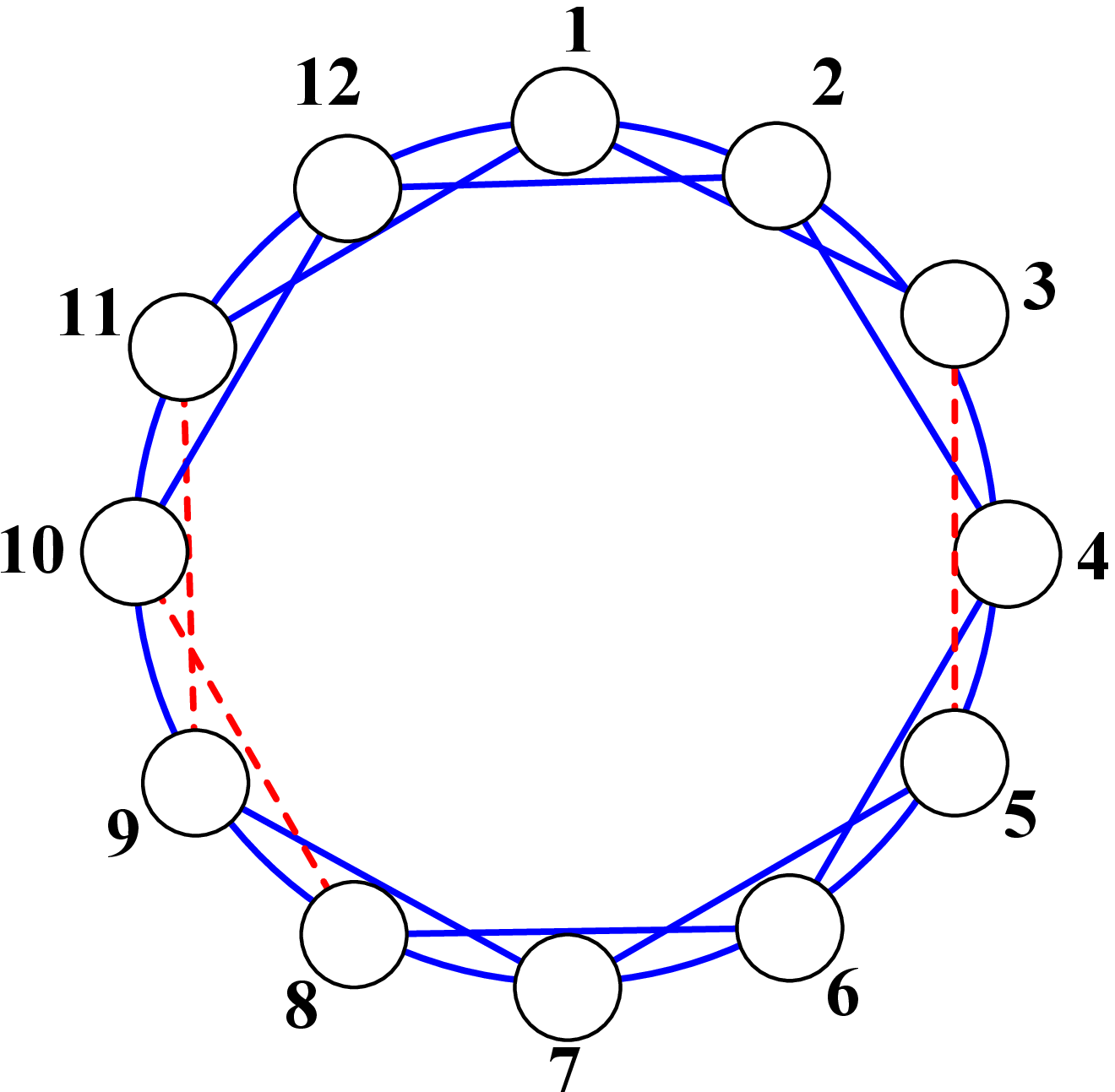}
\\
{\scriptsize (a) Topology of $\mathcal{G}^4$} & {\scriptsize (b) Odd nodes as a hub} & {\scriptsize (c) Measure nodes 2,8 and 10 via hub }& {\scriptsize (d) Delete $h$ long links} \\
\end{tabular}
      \caption{Sparse recovery on graph $\mathcal{G}^4$}
      \label{fig:ring4}
   \end{figure*}

Throughout the paper, given a graph $G=(V,E)$, we say  $S$ forms a \textbf{\textit{hub}} for $U$ if $G_S$ is connected, and for every $u$ in $U$, there exists $s$ in $S$ such that $(u,s) \in E$. Clearly the set of all the odd nodes, denoted by $T_{\textrm{o}}$, form a hub for the set of all the even nodes, denoted by $T_{\textrm{e}}$. Given a $k$-sparse vector $\bfx$, let $\bfx_{\textrm{o}}$ and $\bfx_{\textrm{e}}$ denote the subvectors of $\bfx$ with odd and even indices. Then $\bfx_{\textrm{o}}$ and $\bfx_{\textrm{e}}$ are at most $k$-sparse. The sum of entries in $\bfx_{\textrm{o}}$, denoted by $s_{\textrm{o}}$, can be obtained by one measurement, and similarly for the sum $s_{\textrm{e}}$ of the entries of $\bfx_{\textrm{e}}$. For any subset $W$ of $T_{\textrm{e}}$, $T_{\textrm{o}} \cup W$ induces a connected subgraph and thus can be measured by one measurement. We can obtain the sum of values corresponding to  nodes in $W$ by measuring nodes in $T_{\textrm{o}} \cup W$ and then subtracting $s_{\textrm{o}}$ from the sum.  For example in Fig. \ref{fig:ring4} (b) and (c), in order to measure the sum of the pink nodes 2, 8 and 10, we  measure the sum of pink nodes and all the black odd nodes, and then subtract $s_{\textrm{o}}$ from the obtained summation. Though the subgraph induced by $T_{\textrm{e}}$ are not complete, we can indeed freely measure nodes in $T_{\textrm{e}}$ with the help of the hub $T_{\textrm{o}}$. Therefore $M^C_{k, \lfloor n/2 \rfloor }+1$ measurements are enough to recover $\bfx_{\textrm{e}}\in \mathcal{R}^{\lfloor n/2 \rfloor}$, where the additional one measurement measures $s_{\textrm{o}}$. 
Similarly, we can use $T_{\textrm{e}}$ as a hub to recover the subvector $\bfx_{\textrm{o}} \in \mathcal{R}^{\lceil n/2 \rceil}$ with $M^C_{k, \lceil n/2 \rceil }+1$ measurements, and thus $\bfx$ is recovered. From above, we have
\begin{theorem}\label{thm:ring4}
 All $k$-sparse vectors associated with $\mathcal{G}^4$ can be recovered with $M^C_{k, \lfloor n/2 \rfloor}+M^C_{k, \lceil n/2 \rceil}+2$ measurements, which is $O( 2k\log (n/(2k)))+2$. 
 \end{theorem}
Theorem \ref{thm:ring4} is important in the following three aspects.

Firstly, from ring network to $\mathcal{G}^4$, although the number of links only increases by $n$, the number of measurements required to recover $k$-sparse vectors significantly reduces from $\Theta(n)$ to $O( 2k\log (n/(2k)))+2$. Besides, this value 
is in the same order as $M^C_{k,n}$, while 
the number of links in $\mathcal{G}^4$ is only $2n$ compared with  $n(n-1)/2$ links in a complete graph. 

Secondly, the idea of using a hub to design the measurements is  very important for our later results. If set $S$  can serve as a hub for $U$ in graph $G$, then the induced graph $G_U$ is ``almost equivalent'' to a complete subgraph in the sense that we can measure any subset of nodes in $U$ freely via $S$. The number of measurements required to recover $k$-sparse vectors associated with $U$ is $M^C_{k, |U|}+1$ with one additional one measurement for the hub.

Thirdly, our estimate $O( 2k\log (n/(2k)))+2$ on the minimum number of measurements required to recover $k$-sparse vectors greatly improves over the existing results in \cite{CKMS10,XMT11}, both of which are based on the mixing time of a random walk.
The 
mixing time $T(n)$ is  the smallest $t'$ such that
a random walk of length $t'$ starting at any node in $G$ ends up
having a distribution $\mu'$ with $\|\mu-\mu'\|_{\infty} \leq 1/(2cn)^2$ 
for some $c\geq 1$, where $\mu$ is the stationary distribution over the nodes of a standard random walk
over the graph $G$. 
\cite{XMT11} proves that $O( k T^2(n) \log n)$ measurements 
can identify $k$-sparse vectors with overwhelming probability by compressed sensing.  \cite{CKMS10} uses $O(k^2 T^2(n) \log (n/k))$  measurements to identify $k$ non-zero elements by group testing. In $\mathcal{G}^4$, one can easily see that 
$T(n)$ should be at least $n/4$. Then both results provide no saving in the number of measurements for $\mathcal{G}^4$ as the mixing time is $\Theta(n)$. 

Besides the explicit measurement construction described before Theorem \ref{thm:ring4},  
we can also recover $k$-sparse vectors with $O( \log n)$ random measurements with high
probability. 
We need to point out that 
these random measurements do not depend on the measurements of a complete graph.

Consider an $n$-step Markov chain $\{ X_k, 1 \leq k \leq n\}$ with $X_1=1$. For any $k\leq n-1$, if $X_k=0$, then $X_{k+1}=1$; if $X_k=1$, then $X_{k+1}$ can be 0 or 1 with equal probability.
Clearly any realization of this Markov chain does not contain two or more consecutive
zeros, and thus is a feasible row of the measurement matrix. 
Moreover, 
%
\begin{theorem}\label{thm:random}
With high probability all $k$-sparse vectors associated with $\mathcal{G}^4$ can be recovered with $O(g(k) \log n)$
measurements obtained from the above Markov chain, where $g(k)$ is a function of
$k$.
\end{theorem}

\begin{proof}
See Appendix.
\end{proof}

Adding $n$ links in the
form $(i, i+2 (\textrm{mod } n))$ to the ring network greatly
reduces the number of measurements needed from $\Theta(n)$ to $O(\log n)$. 
Then how many links in the form $(i, i+2(\textrm{mod } n))$ shall we add to the ring network such that the minimum number of measurements required to recover $k$-sparse vectors is exactly $\Theta (\log n)$? The answer is $n- \Theta(\log n)$.
To see this, let $\mathcal{G}^4_h$  denote the graph obtained by deleting $h$ links in the form $(i, i+2 (\textrm{mod } n))$ from
$\mathcal{G}^4$. For example in Fig. \ref{fig:ring4} (d), we delete links $(3,5)$, $(8,10)$ and $(9,11)$ in red dashed lines from $\mathcal{G}^4$. Given $h$, 
our following results do not depend on the specific choice of links to remove.  We have

\begin{theorem}\label{thm:ring4h}
 The minimum number of measurements required to recover $k$-sparse vectors associated with $\mathcal{G}^4_h$ is lower bounded by $\lceil h/2 \rceil$, and upper bounded by $2M^C_{k,\lceil \frac{n}{2}\rceil}+h+2$.
 \end{theorem}
%
%
\begin{proof}
Let $D$ denote the set of nodes such that for every $i \in D$, link $(i-1, i+1)$
is removed from $\mathcal{G}^4$.
The proof of the lower bound follows the proof of Theorem 2 in \cite{TWHR11}.
The key idea is that recovering one non-zero element in $D$ is equivalent to recovering one non-zero element in a ring network with $h$ nodes, and thus $\lceil h/2 \rceil$ measurements are necessary.

For the upper bound, we 
first measure nodes in $D$ separately with $h$ measurements.  Let $S$ contain the even nodes in $D$ and all the odd nodes. 
$S$ can be used as a hub to recover the $k$-sparse subvectors associated with the even nodes that are not in $D$, and the number of measurements used is at most $M^C_{k,\lfloor \frac{n}{2}\rfloor}+1$. We similarly recover $k$-sparse subvectors associated with odd nodes that are not in $D$ using the set of the odd nodes in $D$ and all the even nodes as a hub.  The number of measurements is at most $M^C_{k,\lceil \frac{n}{2}\rceil}+1$. Sum them up and the upper bound follows.
%
%
\end{proof}

Together with (\ref{eqn:MC}), Theorem \ref{thm:ring4h} directly implies that
 if $\Theta(\log n)$ links in the form $(i, i+2 (\textrm{mod } n))$ are deleted from $\mathcal{G}^4$, then $\Theta( \log n)$ measurements are both necessary and sufficient to recover $k$-sparse vectors associated with $\mathcal{G}^4_{\Theta(\log n)}$ for any constant $k$. Moreover, the lower bound in Theorem \ref{thm:ring4h} implies that if the number of links removed is $\Omega( \log n)$, then the number of measurements required for sparse recovery is also $\Omega( \log n)$. Thus, we need to add $n-\Theta(\log n)$ links to a ring network such that the number of measurements required for sparse recovery is exactly $\Theta(\log n)$.

Since the number of measurements required by compressed sensing is greatly reduced when we add $n$ links to the ring network, one may wonder whether
the number of measurements needed to locate $k$ non-zero elements by group
testing can also be greatly reduced or not. Our next result shows that
this is not the case for group testing.

\begin{prop}\label{thm:ring4gt}
$\lfloor n/4\rfloor$ measurements are necessary to locate two non-zero elements associated with $\mathcal{G}^4$ by group testing.
\end{prop}

\begin{proof}
Suppose two non-zero elements are on nodes
$2i-1$ and $2i$ for some $1\leq i \leq \lfloor \frac{n}{2}\rfloor$.
We view nodes $2i-1$ and $2i$ as a group for every $i$ ($1
\leq i \leq \lfloor \frac{n}{2}\rfloor$), denoted by $B_i$. 
If both nodes in $B_j$ are `1's for some $j$, then every measurement that
passes either node or both nodes in $B_i$ is always `1'.  
Consider a reduced graph with
$B_i$, $\forall i$ 
as nodes, and
   link $(B_i,B_j)$ ($i \neq j$) exists only if in $\mathcal{G}^4$ there is a path from a node in $B_i$ to a node in $B_j$
without going though any other node not in $B_i$ or $B_j$. 
The reduced network is a ring with $\lfloor \frac{n}{2}\rfloor$ nodes, 
and thus $\lfloor n/4 \rfloor$ measurements are required to locate one
non-zero element in the reduced network. 
Then the lower bound follows.
\end{proof}

By Theorem \ref{thm:ring4} and Proposition \ref{thm:ring4gt}, we
observe that in $\mathcal{G}^4$, with compressed sensing the number
of measurements needed to recover $k$-sparse vectors is $O(2k \log (n/(2k)))$,
while with group testing, $\Theta(n)$ measurements are required if
$k \geq 2$. 

\subsection{Two-dimensional grid}\label{sec:2d}
Next we consider the two-dimensional grid, denoted by $\mathcal{G}^{2d}$. 
$\mathcal{G}^{2d}$ has $\sqrt{n}$ rows and $\sqrt{n}$ columns. From now on we skip `$\lceil\cdot \rceil$' and `$\lfloor\cdot \rfloor$' for notational simplicity, but note that the number of nodes should always be an integer.

\begin{figure}
\centering
\begin{tabular}{c c}
\includegraphics[width=0.35\linewidth]{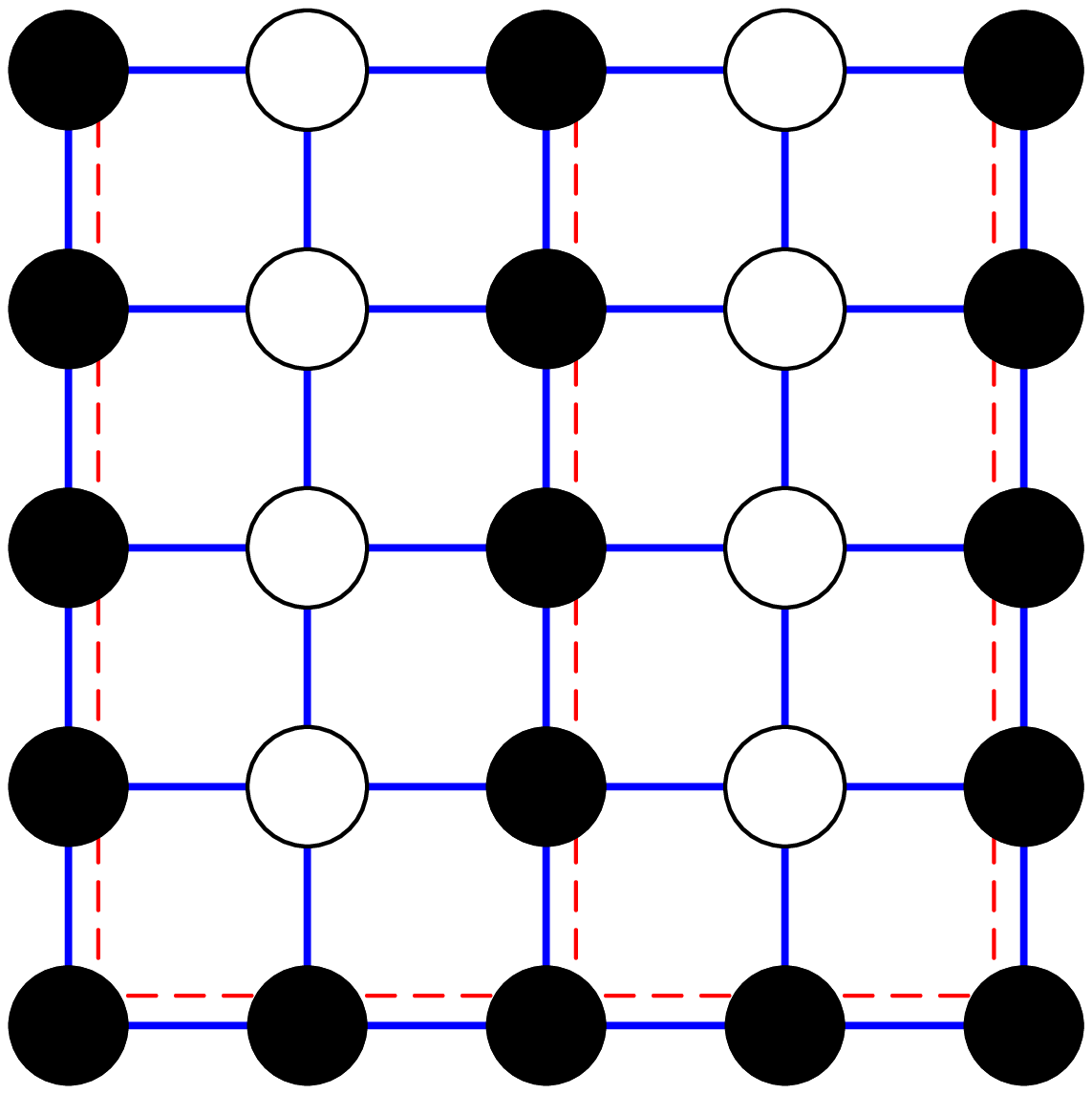}
&
\includegraphics[width=0.35\linewidth]{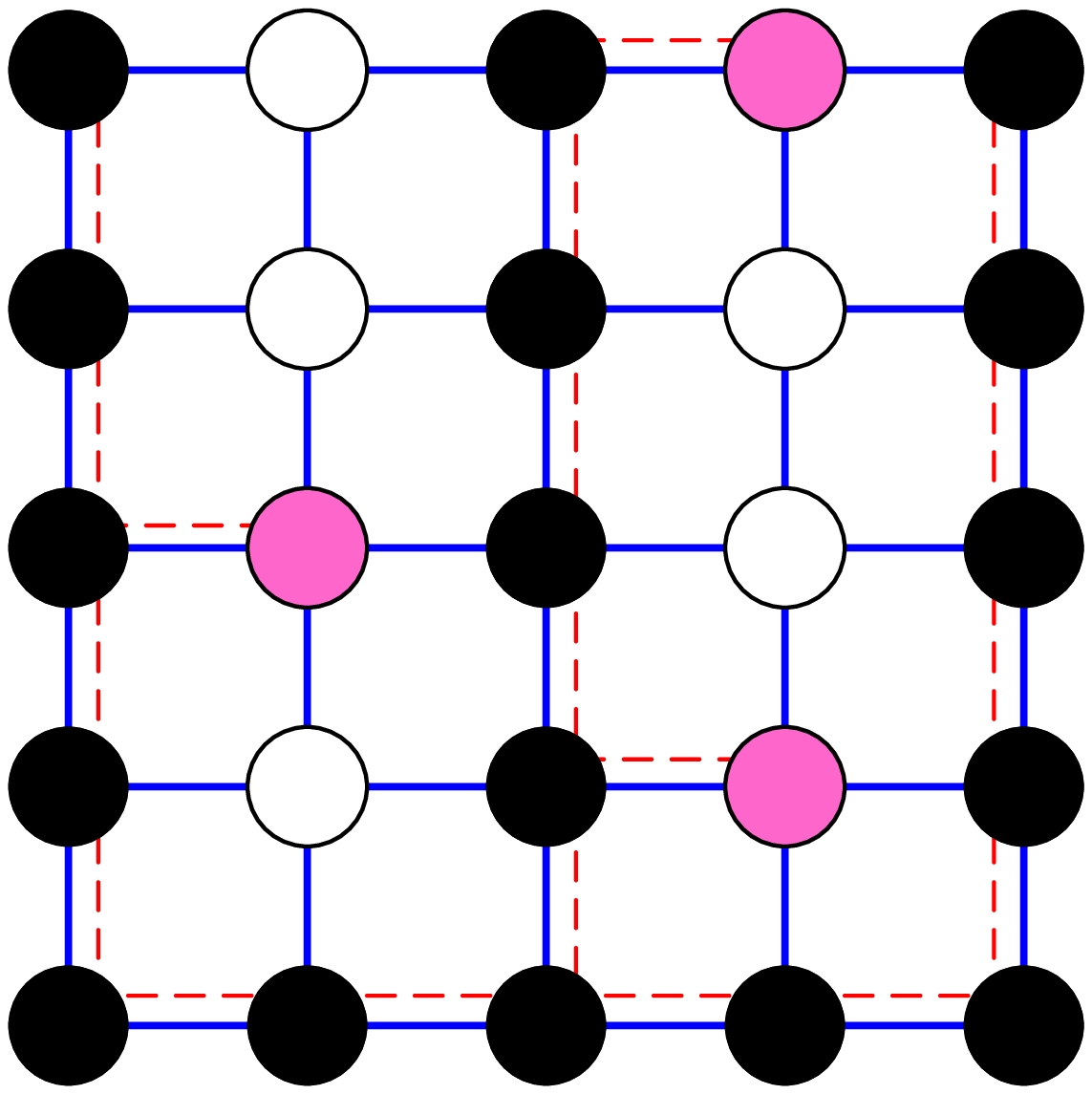}
\\
{\scriptsize (a) The set of black nodes as a hub } & {\scriptsize (b) Measure pink nodes via the hub}\\
\end{tabular}
      \caption{Sparse recovery on two-dimensional grid}
      \label{fig:2d}
   \end{figure}

We assume $\sqrt{n}$ to be even here for notational simplicity, and the result can be easily modified for the case that $\sqrt{n}$ is odd.
The idea of measurement construction is similar to that for graph $\mathcal{G}^4$.
First, Let $S_1$ contain the nodes in the first row and all the nodes in the odd columns. Then $S_1$ can be used as a hub to measure $k$-sparse subvectors associated with nodes in $V \backslash S_1$, as shown in Fig. \ref{fig:2d}. The number of measurements 
is $M^C_{k, (n/2-\sqrt{n}/2)}+1$. Then let $S_2$ contain the nodes in the first row and all the nodes in the even columns, and use $S_2$ as a hub to recover up to $k$-sparse subvectors associated with nodes in $V \backslash S_2$. Then number of measurements required is also $M^C_{k, (n/2-\sqrt{n}/2)}+1$. Finally, use nodes in the second row as a hub to recover sparse subvectors associated with nodes in the first row. Since nodes in the second row are already identified in the above two steps, then we do not need to measure the hub separately in this step. The number of measurements here is  $M^C_{k, \sqrt{n}}$. 
Therefore,
\begin{theorem}
The number of measurements needed to recover $k$-sparse vectors associated with $\mathcal{G}^{2d}$ is at most $2M^C_{k, n/2-\sqrt{n}/2}+M^C_{k,\sqrt{n}}+2$. 
\end{theorem}

\subsection{Tree}\label{sec:tree}
Next we consider a tree topology as in Fig. \ref{fig:tree}. For a given tree, the
root is treated as the only node in layer 0. The nodes that are $t$
steps away from the root are in layer $t$. We say the tree has depth
$h$ if the farthest node is $h$ steps away from the root. Let $n_i$ denote the number of nodes on layer $i$, and $n_0=1$. We 
construct measurements to recover vectors associated with a tree by the following \textit{tree approach}.
\begin{figure}[h]
\begin{center}
  \includegraphics[scale=0.25]{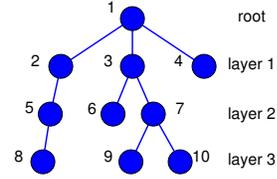}
  \caption{Tree topology}\label{fig:tree}
  \end{center}
\end{figure}
%

We recover the nodes layer by layer
starting from the root, and recovering nodes in layer $i$ requires
that all the nodes above layer $i$ should already be recovered. First measure the root separately.  When recovering the subvector associated with nodes in  layer $i$ ($2\leq i \leq h$), 
we can measure the sum of any subset of nodes in layer $i$ using some nodes in the upper layers as hub and then delete the value of the hub from the obtained sum.
One simple way to find a hub 
is to trace back from  nodes to be
measured
on the tree simultaneously until they reach one same node. For example in Fig. \ref{fig:tree}, 
in order to measure nodes 5 and 7 together, we will trace back to the root and measure nodes 1, 2, 3, 5, and 7 together and then subtract the values of nodes 1, 2, and 3, which are already identified when we recover nodes in the upper layers. 
With this approach, we have,
\begin{theorem}\label{thm:tree}
$\sum_{i=0}^h M^C_{k,n_i}$ measurements are enough to recover $k$-sparse vectors associated with a tree with depth $h$, where $n_i$ is the number of nodes in layer $i$. 
\end{theorem}
%

\section{Sparse Recovery over General Graphs}\label{sec:general}

In this section we consider recovering $k$-sparse vectors over general graphs. The graph is assumed to
be connected. If not, we simply  treat each component as a connected subgraph and
design measurements to recover $k$-sparse subvectors associated with each subgraph separately. 

Inspired by the construction methods in Section \ref{sec:special}, 
in Section \ref{sec:bound}
we propose a general design guideline based on ``$r$-partition'' which will be introduced soon. The key idea is to divide the nodes into a small number of groups such that nodes in the same group are connected to one hub, and thus 
can be measured freely with the help of the hub. We use the  Erd\H{o}s-R\'enyi random graph as an example to illustrate the design guideline based on $r$-partition. Since finding the minimum number of such groups in general turns out to be NP-hard, in Section \ref{sec:algo} we propose a simple algorithm to design a small number of measurements to recover $k$-sparse vectors associated with any given graph.

%
%
%
%
\subsection{Measurement Construction Based on $r$-partition}\label{sec:bound}

   In $\mathcal{G}^4$, we divide nodes into odd nodes $T_{\textrm{o}}$ and even nodes $T_{\textrm{e}}$ and use each set as a hub 
   for the other set. In general graphs, we extend this idea 
   and have the following definition:
 \begin{defi}[$r$-partition]\label{def:rp}
Given $G=(V,E)$, disjoint subsets $N_i$ ($i=1,...,r$) of $V$ form an \textbf{$r$-partition} of $G$ if and only if these two conditions both hold: (1) $\cup_{i=1}^r N_i=V$, and (2) $\forall i$, $V \backslash N_i$ is a hub for $N_i$.
 \end{defi}
Clearly, $T_{\textrm{o}}$ and $T_{\textrm{e}}$ form a $2$-partition of graph $\mathcal{G}^4$. With the above definition, we have 
\begin{theorem}\label{thm:component}
If $G$ has an $r$-partition $N_i$ ($i=1,...,r$), then the number of measurements needed to recover $k$-sparse vectors associated with $G$ is at most  $\sum_{i=1}^r M^C_{k,|N_i|}+r$, which is $O(r k\log (n/k))+r$. 
\end{theorem}
\begin{proof}
Note that  $M^C_{k, |N_i|}+1$ measurements (with one additional measurement for $V \backslash N_i$) are enough to recover $k$-sparse subvector associated with $N_i$  via its hub $V \backslash N_i$.
\end{proof}
We next apply this result to the Erd\H{o}s-R\'enyi random graph $G(n,p)$, which contains $n$ nodes and there exists an link between any two nodes independently with probability $p$.  
Note that if $p \geq (1+\epsilon)\log n/n$ for some constant $\epsilon >0$, $G(n,p)$ is connected almost surely \cite{Bollobas01}.
\begin{theorem}
For Erd\H{o}s-R\'enyi random graph $G(n,p)$ with $p=\beta\log n/n$, if $\beta \geq 2+\epsilon$ for some constant $\epsilon>0$, then any two disjoint subsets $N_1$ and $N_2$ of nodes with $|N_1|=|N_2|=n/2$ form a 2-partition with high probability. Moreover, with high probability the number of measurements needed to recover $k$-sparse vectors associated with $G(n,p)$ is at most $2M^C_{k, n/2}+2$, which is $O(2k \log(n/(2k)))+2$.
\end{theorem}
\begin{proof}
Let $N_1$ be any subset of $V$ with $|N_1|=n/2$, and let
$N_2=V \backslash N_1$.  
Then $G_{N_1}$ and $G_{N_2}$ are both Erd\H{o}s-R\'enyi random graphs with $n/2$ nodes, and are connected almost surely when $p\geq (2+\epsilon) \log n/n$. 

We claim that with high probability, for every $u \in
N_1$, there exists $v \in N_2$ such that $(u,v) \in E$. 
Let $P_1$ denote the probability that there exists some $u \in
N_1$ such that $(u, v) \notin E$ for every $v \in N_2$. Then
\begin{align*}
P_1&= \sum_{u\in N_1}(1-p)^{n/2}= \frac{n}{2}
 (1-\beta \log n/n)^{n/2}\\
  &= \frac{n}{2}
(1-\frac{\beta \log n}{n})^{ \frac{n}{\beta \log n}\cdot \frac{ \beta \log n}{2}}\leq \frac{n}{2} e^{-\frac{\beta\log n}{2}} \leq \frac{n^{-\epsilon/2}}{2}, 
\end{align*}
where the last inequality holds from $\beta \geq 2+\epsilon$. Then $P_1$ goes to zero as $n$ goes to infinity, 
and the claim follows.  Similarly, one can prove that with
high probability for every $v \in N_2$, there exists $u
\in N_1$ such that $(u,v) \in E$.

Then with high probability $N_1$ and $N_2$ form a 2-partition. The second statement follows from Theorem \ref{thm:component}.
\end{proof}

\cite{CKMS10} considers group testing over Erd\H{o}s-R\'enyi random graphs and shows that $O(k^2 \log^3 n)$ measurements 
are enough to identify up to $k$ non-zero entries in an $n$-dimensional logical vector provided that $p=\Theta (k \log^2 n/n)$. Here with compressed sensing setup and 2-partition results, we can recover $k$-sparse vectors in $\mathcal{R}^n$ with $O(2k\log(n/(2k)))+2$ measurements when $p>(2+\epsilon)\log n/n$ for some $\epsilon>0$. Note that this result also improves over the previous result in \cite{XMT11}, which requires $O(k \log^3 n)$ measurements for compressed sensing on $G(n,p)$.

From Theorem \ref{thm:component}, the number of measurements used is closely related to the value $r$. In general one  wants to reduce $r$ so as to reduce the number of measurements. 
Given graph $G$ and integer $r$, the question that whether or not $G$ has an $r$-partition is called \textit{$r$-partition problem}.
In fact,
\begin{prop}\label{thm:np}
$\forall r \geq 3$, $r$-partition problem is NP-complete.
\end{prop}
Please refer to Appendix for its proof. We remark that we cannot prove the hardness of the $2$-partition problem though we conjecture it is also a hard problem.
\subsection{Measurement Construction Algorithm for General Graphs}\label{sec:algo}
Section \ref{sec:bound} proposes the $r$-partition concept as a measurement design guideline. 
But finding an $r$-partition with the smallest $r$ in general is NP-hard. 
Now given a connected graph $G$, how shall we efficiently design a small number of measurements 
to recover $k$-sparse vectors associated with $G$?

One simple way is to find the spanning tree of $G$, and then use the tree approach in Section \ref{sec:tree}. 
The depth of the spanning tree is at least $R$, where $R=\min_{u \in V}
\max_{v \in V} d_{uv}$ is the radius of $G$ with $d_{uv}$ as the
length of the shortest path between $u$ and $v$. 
This approach only uses links in the spanning tree, and the number of measurements used 
is large when the radius $R$ is large. For example, the radius of  $\mathcal{G}^4$ in Fig. \ref{fig:ring4} is $n/4$, then the spanning tree approach uses at least $n/4$ measurements, one for each layer. However,  
the number of measurements can be as small as $O(2k \log (n/2k))+2$ if we take advantage of the additional links.

Here we propose a simple algorithm to design the measurements for general graphs. The algorithm combines the ideas of the tree approach and the $r$-partition. 
We still  divide nodes into a small number of groups such that each group can be identified via some hub. Here nodes in the same group are the leaf nodes of a spanning tree of a gradually reduced graph. A leaf node has no
children on the tree.

Let $G^*=(V^*,E^*)$ denote the input graph. The algorithm is built on the following two subroutines. 
 \textbf{Leaves}($G$, $u$) returns the set of leaf nodes of a spanning tree of $G$ rooted at $u$.  \textbf{Reduce}($G=(V,E)$, $u$, $H$) deletes $u$ from $G$ and fully connects all the neighbors of $u$. Specifically, for every two neighbors $v$ and
$w$ of $u$, 
we add a link
$(v,w)$, if not already exist, 
and let $H_{(v,w)}=H_{(v,u)}\cup H_{(u,w)} \cup \{u \}$, where for each link $(s,t) \in E$, $H_{(s,t)}$ denotes the set of nodes, if any,  that serves as a hub for $s$ and $t$ in the original graph $G^*$. We record $H$ such that measurements constructed on a reduced graph $G$ can be feasible in $G^*$.

%
 \floatname{algorithm}{Subroutine}
\begin{algorithm}
\begin{algorithmic}[1]
\REQUIRE graph $G$, root $u$
 \STATE  Find a spanning tree $T$ of $G$ rooted at $u$ by breadth-first search, and let $S$ denote the set of leaf nodes of $T$.
 \RETURN $S$
 \end{algorithmic}
 \caption{\textbf{Leaves}($G$, $u$)}
 \end{algorithm}
 \begin{algorithm}
\begin{algorithmic}[1]
\REQUIRE  $G=(V, E)$,  $H_e$ for each $e \in E$, and node $u$
 \STATE  $V=V\backslash u$.
 \FORALL {two different neighbors $v$ and
$w$ of $u$}
 \IF {$(v,w) \notin E$}
 \STATE  $E=E \cup (v,w)$, $H_{(v,w)} =H_{(v,u)} \cup H_{(u,w)} \cup \{u \}$.
 \ENDIF
 \ENDFOR
 \RETURN $G$, $H$
 \end{algorithmic}
 \caption{\textbf{Reduce}($G$, $u$, $H$)}
 \end{algorithm}
\setcounter{algorithm}{0}
 \floatname{algorithm}{Algorithm}
 \renewcommand{\algorithmicreturn}{\textbf{Output:}}
  \begin{algorithm}
\begin{algorithmic}[1]
\REQUIRE  $G^*=(V^*, E^*)$.
 \STATE $G=G^*$, $H_e=\varnothing$ for each $e \in E$
\WHILE { $|V| > 1$}
 \STATE  Find the node $u$ such that $\max_{v \in V} d_{u v}= R^G$, where $R^G$ is the radius of $G$.  $S=$\textbf{Leaves}($G$, $u$).
  \STATE Design $f(k, |S|)+1$ measurements to recover $k$-sparse vectors associated with $S$ using nodes in $V\backslash S$ as a hub.
   \FORALL {$u$ in $S$}
 \STATE $G=$ \textbf{Reduce}($G$, $u$, $H$)
 \ENDFOR
 \ENDWHILE
 \STATE Measure the last node in $V$ directly.
 \RETURN All the measurements.
 \end{algorithmic}
 \caption{Measurement construction for graph $G^*$}\label{algo:design}
 \end{algorithm}

Given graph $G^*$, 
let $u$ denote the
node such that $\max_{v \in V^*} d_{u v}= R$, where $R$ is the radius of $G^*$. Pick $u$ as the root
and obtain a spanning tree $T$ of $G^*$ by breadth-first search. Let
$S$ denote the set of leaf nodes in $T$.  
With $V^* \backslash S$ as a hub, we
can design $f(k,|S|)+1$ measurements to recover up to $k$-sparse vectors associated with $S$.  
We then reduce the network by deleting every $u$ in $S$ and fully connects all the neighbors of $u$. 
For the obtained reduced network $G$, we repeat the above process until all the nodes are deleted. 
Note that when designing the measurements in a reduced graph $G$, if a measurement uses link $(v, w)$, then it should also include nodes in $H_{(v,w)}$ so as to be feasible in the original graph $G^*$.

In each step tree $T$ is rooted at node $u$ where $ \max_{v \in {V}} d_{u v}$ equals the radius of the current graph $G$. Since all the leaf nodes
of $T$ are deleted in the graph reduction procedure, the radius of the new obtained graph should be reduced by at least one. 
Then we have at most $R$ iterations in Algorithm \ref{algo:design} until only one node is left. Clearly we have,
\begin{prop}\label{prop:algo}
The number of measurements designed by Algorithm \ref{algo:design} is at most $Rf(k, n)+R+1$, where $R$ is the radius of the graph.
\end{prop}
We  remark  that the number of measurements by the spanning tree approach we mentioned at the beginning of Section \ref{sec:algo} is also no greater than $Rf(k, n)+R+1$.   
However, we expect that Algorithm 1 uses fewer measurements than the spanning tree approach for general graphs, since Algorithm 1 also considers links that are not in the spanning tree. And it is verified in Experiment 1 in Section \ref{sec:simu}. 

\section{Simulation} \label{sec:simu}
\begin{figure*}[ht]
\begin{minipage}{2.3in}
\begin{tabular}{c}
\begin{psfrags}
\psfrag{Upper bound of number of measurements}[bl][bl][1.1]{\tiny Upper bound of number of measurements}
\psfrag{Number of measurements}[bl][bl][1.1]{\tiny Number of measurements}
\psfrag{Radius}[bl][bl][1.1]{\tiny Radius}
\psfrag{Number of links}[cl][l][1.2]{\tiny Number of links}
\includegraphics[width=1.1\linewidth,height=0.9\linewidth]{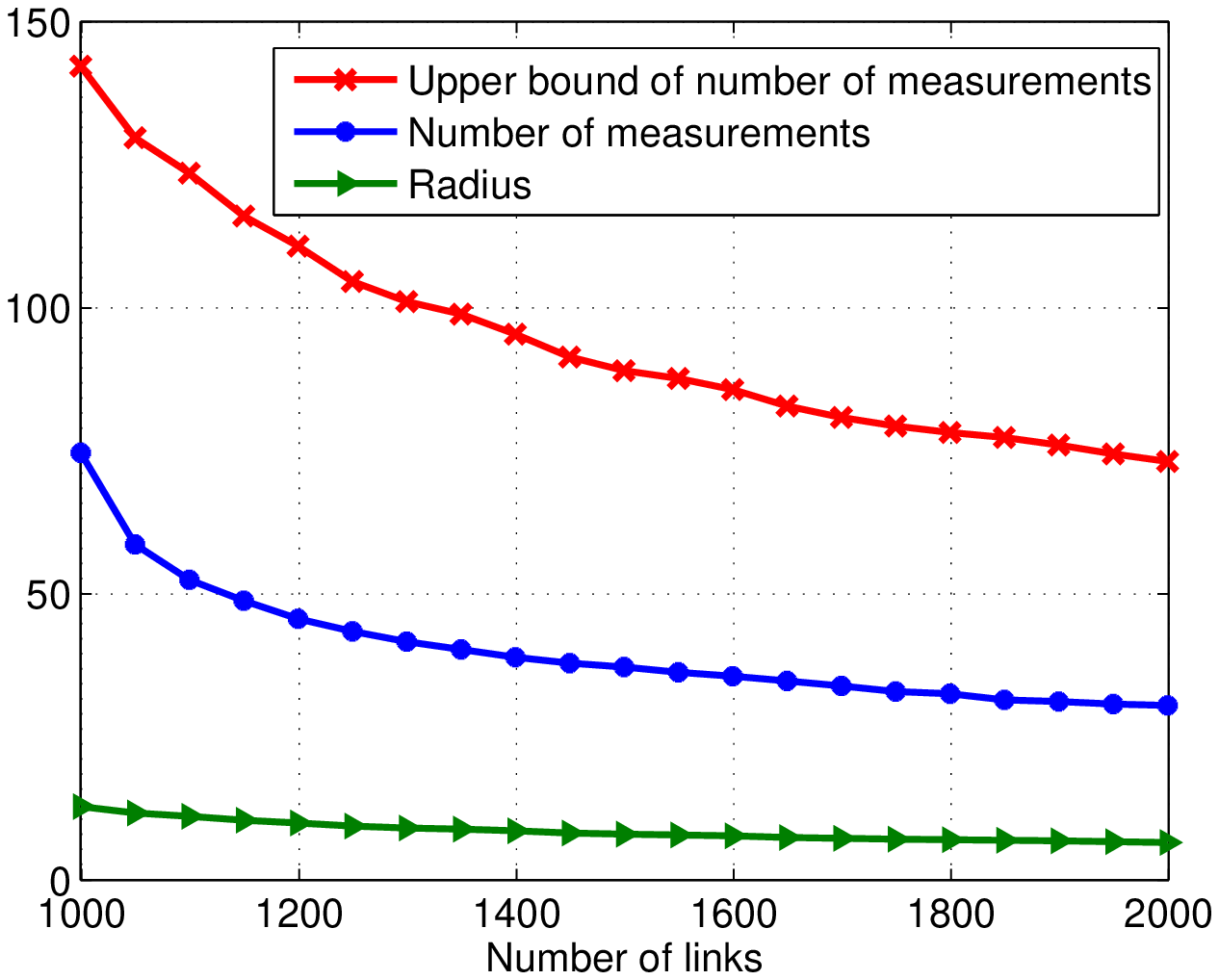}\label{fig:treeplus}
\end{psfrags}
\\
{\scriptsize Fig. 6. Random graph with $n=1000$}
\end{tabular}
\end{minipage}
\begin{minipage}{2.3in}
\begin{tabular}{c}
\begin{psfrags}
\psfrag{m=1 m=1}[bl][bl][1.1]{\tiny $m=1$}
\psfrag{m=2 m=2}[bl][bl][1.1]{\tiny $m=2$}
\psfrag{m=3 m=3}[bl][bl][1.1]{\tiny $m=3$}
\psfrag{Number of nodes}[bl][l][1.2]{\tiny Number of nodes}
\psfrag{Number of measurements}[cl][l][1.2]{\tiny Number of measurements}
\includegraphics[width=1.1\linewidth,height=0.9\linewidth]{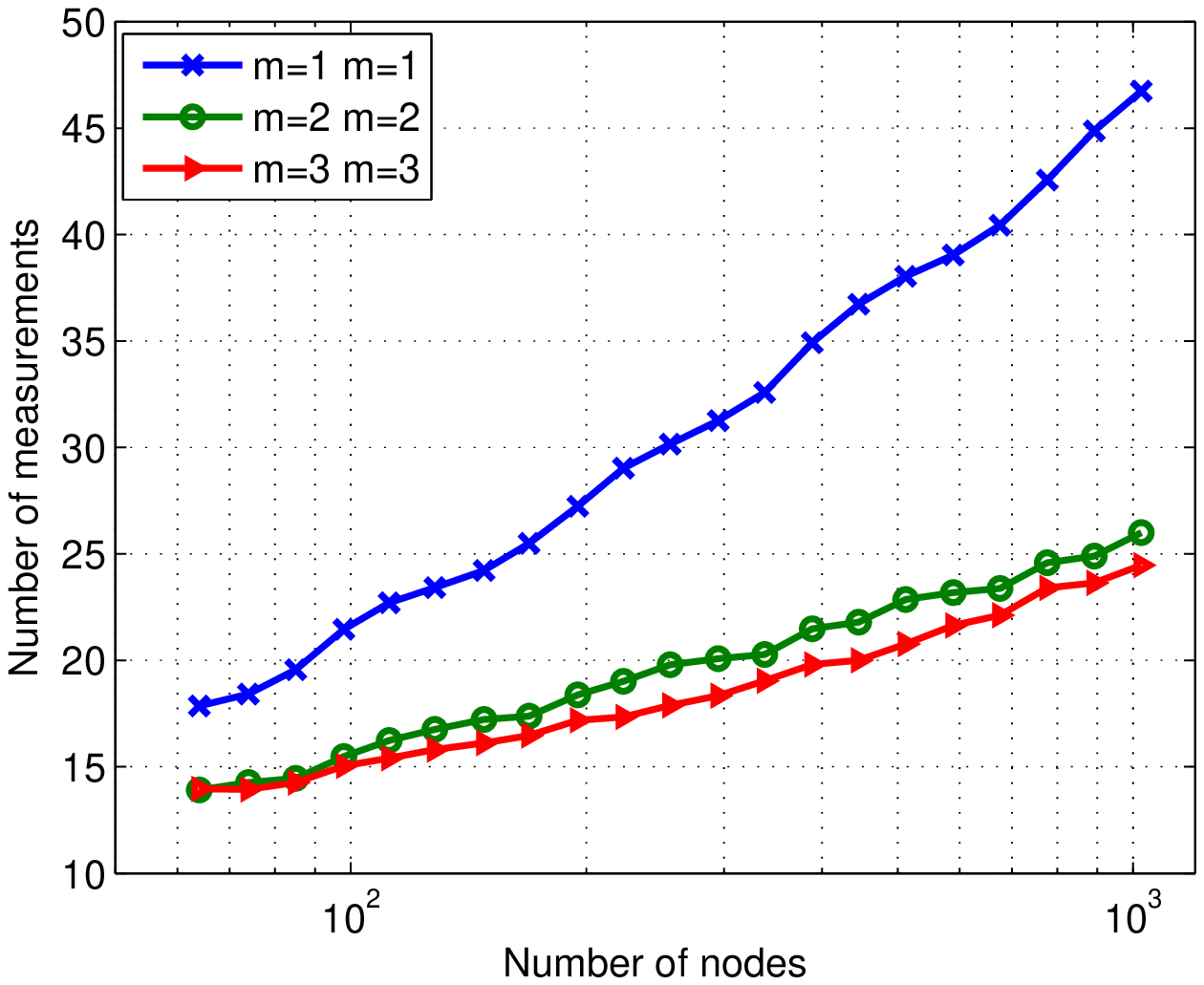}\label{fig:BA}
\end{psfrags}\\
{\scriptsize Fig. 7. BA model with increasing $n$ and different $m$}
\end{tabular}
\end{minipage}
\begin{minipage}{2.3in}
\begin{tabular}{c}
\begin{psfrags}
\psfrag{ellell, with noise noise}[bl][bl][1.1]{\tiny $\ell_1$, with noise}
\psfrag{ellell, no noise noise}[bl][bl][1.1]{\tiny $\ell_1$, no noise}
\psfrag{Ours, with noise noise}[bl][bl][1.1]{\tiny Ours, with noise}
\psfrag{Ours, no noise noise}[bl][bl][1.1]{\tiny Ours, no noise}
\psfrag{Support size of the vectors}[cl][l][1.2]{\tiny Support size of the vectors}
\psfrag{xrxo2}[bc][l][1]{\tiny $\|\bfx_{\textrm{r}}-\bfx_0\|_2$}
\includegraphics[width=1.1\linewidth,height=0.9\linewidth]{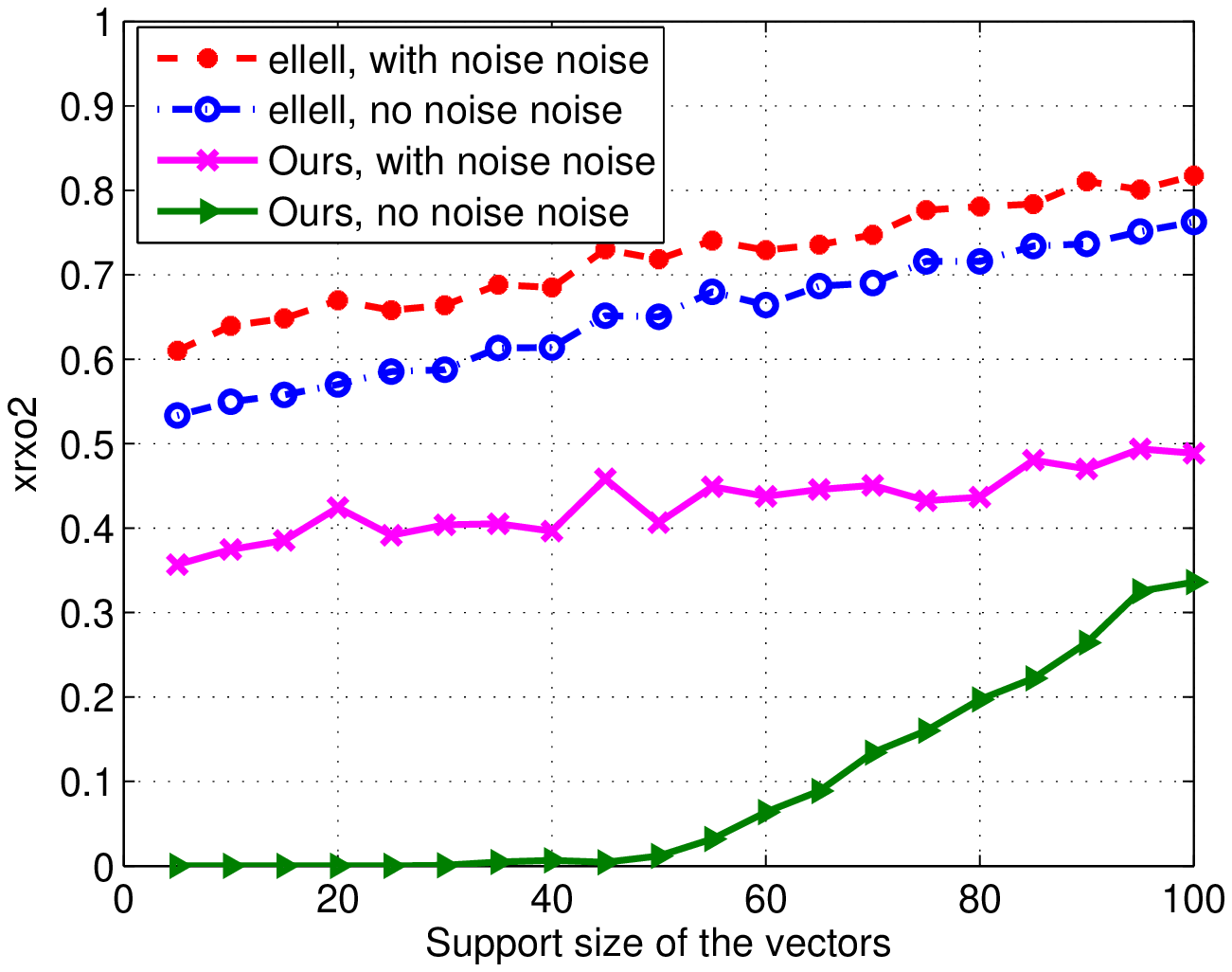}\label{fig:simunoise}
\end{psfrags}\\

{\scriptsize Fig. 8. Recovery performance with hub errors}
\end{tabular}
\end{minipage}
\end{figure*}

\noindent\textbf{{Experiment 1 (Effectiveness of Algorithm 1):}} Given a graph $G$, we consider recovering $1$-sparse vectors associated with $G$. 
Note that $M^C_{1,n}=\lceil \log (n+1) \rceil$ 
and the corresponding measurement matrix has the binary expansion of $i$ as column $i$ \cite{Dorfman43}. Algorithm 1  divides the nodes into groups such that each group (except the last one) can be measured freely via some hub. 
The last group only contains one node and can be measured directly. The total number of measurements by Algorithm 1 is $\sum_i^{q-1} \lceil \log (n_i+1)\rceil+q$, where $n_i$ is the number of nodes in group $i$ and $q$ is the total number of groups.

In Fig. 6, we gradually increase the number of links in a graph with $n=1000$ nodes.  We start with a uniformly generated random tree, and in each step randomly add $25$ links that do not already exist. 
All the results are averaged over one hundred realizations. The number of measurements constructed decreases from 73 to 30 when the number of links increases from $n-1$ to $2n-1$. 
Note that the number of measurements is already within $3 M^C_{1, n}$ when the average node degree is close to 4. The radius of the graph decreases from 13 to 7, and we also plot the upper bound 
in Proposition \ref{prop:algo}. One can see that 
the number of measurements constructed can be much less than the upper bound.

In Fig. 7, we consider the scale-free network with Barab\'asi-Albert (BA) model \cite{BA99} where the graph initially has $m_0$ connected nodes, and each new node connects to $m$ existing nodes with a probability that is proportional to the degree of the existing nodes. We start with a random tree of 10 nodes and increase the total number of nodes from 64 to 1024. Every result is averaged over one hundred realizations. One can see that the number of measurements constructed is proportional to $\log n$, and decreases when $m$ increases.

\vspace{0.1in}

\noindent\textbf{{Experiment 2 (Sparse Recovery Performance with Noise):}} 
Compressed sensing theory indicates that if $A$ is a random 0-1 matrix, 
with overwhelming probability we can recover the sparse vector $\bfx_0$ though $\ell_1$-minimization \cite{CaT05}. 
Here we generate a graph with $n=500$ nodes from BA model. Algorithm 1 divides nodes into four groups with 375, 122, 2 and 1 node respectively. 
For each of the first two groups with size $n_i$ ($i=1,2$), we generate $\lceil n_i/2\rceil$ random measurements each  measuring a random subset of the group together with its hub. We also measure the two hubs directly. 
Each of the three nodes in the next two groups is measured directly by one measurement. The generated matrix $A$ is 254 by 500. We generate a sparse vector $\bfx_0$ with i.i.d. zero-mean Gaussian entries on a randomly chosen support, and normalize $\|\bfx_0\|_2$ to 1. To recover $\bfx_0$ from $\bfy=A\bfx_0$, one can run $\ell_1$-minimization to recover the subvectors associated with the first two groups, and the last three entries of $\bfx_0$ can be obtained from measurements directly. However, note that every measurement for the first two groups passes through its hub, then any error in a hub measurement will affect 
 every measurement for the group of nodes using this hub.
To address this issue, we propose to use a modified $\ell_1$-minimization in which the errors in the two hubs are treated as entries of an augmented vector to recover. Specifically, let the augmented vector $\bfz=[\bfx_0^T ,e_1, e_2]^T$ and the augmented matrix $A'=[A  \  \bm\beta \  \bm \gamma]$, where $e_1$ (or $e_2$) denotes the error in the measurement of the first (second) hub, and the column vector $\bm \beta$ (or $\bm \gamma$) has  `1' in the row corresponding to the measurement of the first (or second) hub and `0' elsewhere. We then recover $\bfz$ (and thus $\bfx_0$) from $\bfy=A'\bfz$ via $\ell_1$-minimization on each group. 
Fig. 8 compares the recovery performance of our modified recovering method and the traditional $\ell_1$-minimization, where the hub errors $e_1$ and $e_2$ are drawn from a Gaussian distribution with zero mean and unit variance. For every support size $k$, we randomly generate one hundred $k$-sparse vectors $\bfx_0$, and let $\bfx_{\textrm{r}}$ denote the recovered vector. Even with the hub errors, the average $\|\bfx_{\textrm{r}}-\bfx_0\|_2$ is within $10^{-6}$ when $\bfx_0$ is at most 25-sparse by our method, while by $\ell_1$-minimization, the value is at least 0.5. We also consider the case that besides errors in  hub measurements,  every other measurement  has i.i.d. zero-mean Gaussian noise.  Let $\bfw$ denote the noise vector and $\|\bfw\|_2$ is normalized to 2.  The average $\|\bfx_{\textrm{r}}-\bfx_0\|_2$ here is smaller with our method than that with $\ell_1$-minimization. 

\section{Conclusion}\label{sec:con}

This paper addresses the sparse recovery problem with graph constraints. 
By providing explicit measurement constructions for different graphs, we derive upper bounds of the minimum number of measurements needed to recover vectors up to certain sparsity. It would be interesting to explore corresponding tight lower bounds. Further efforts are also needed to empirically evaluate the performance of different recovery themes, especially when the measurements are noisy.

\appendix

\subsection{Proof of Theorem \ref{thm:random}}

Let $A^{m
\times n}$ denote the  matrix with $m$ realizations
of the $n$-step Markov chain. To prove the statement, from \cite{CaT05}, we only need
to show that  
the probability that every $2k$ columns of $A$ are linearly independent goes to 1
as $n$ goes to infinity.

Let $A_I$ be a submatrix of $A$ with columns in $I$, where $I$ is an index set with $|I|=2k$.
Let $A_{S_jI}$ ($1\leq j \leq \lfloor \frac{m}{2k}\rfloor$) be a
submatrix of $A_I$ formed by row $2k(j-1)+1$ to row $2kj$ of $A_I$.
Let $P_d^I$ denote the probability that rank($A_I$)$<2k$, and let $\pi _d^I$ denote the probability that rank($A_{S_jI}$)$<2k$
for given $j$. Note that given $I$, $\pi _d^I$ is the same for every
$A_{S_jI}$, $\forall j$. 
Note that rank($A_I$)$< 2k$ implies that rank($A_{S_jI}$)$<2k$ for each
such matrix $A_{S_jI}$, then 
\begin{equation}\label{eqn:pdi}
P_d^I \leq (\pi _d^I)^{\lfloor \frac{m}{2k}\rfloor}.
\end{equation}

To characterize $\pi _d^I$,
 consider matrix $B^{2k \times 2k}$ with $B_{ii}=0$ for
$i=2,3,...,2k$ and $B_{ij}=1$ for all the other elements.
Since rank($B$)$=2k$, then 
\begin{equation}\label{eqn:pidi}
\pi_d^I \leq 1- P(A_{S_jI} \textrm{ is a row permutation of } B).
\end{equation}
One can check that in this Markov chain, for every $ 1 \leq i < k \leq n$, $P(X_k=1~|~X_i=1) \geq 1/2$, $P(X_k=0~|~X_i=1)\geq 1/4$, $P(X_k=1~|~X_i=0) \geq 1/2$,
and $P(X_k=1)\geq 1/2$ by simple calculation.
Since $B$ has $(2k)!$ different row permutations, one can calculate that 
\begin{equation}\label{eqn:pab}
P(A_{S_jI} \textrm{ is a row permutation of } B)\geq
(2k)!/2^{4k^2+2k-1}.
\end{equation}
Combining (\ref{eqn:pdi}), (\ref{eqn:pidi}) and (\ref{eqn:pab}), we
have
\begin{align}
&P(\textrm{every } 2k \textrm{ columns of } A \textrm{ are
linearly independent}) \nonumber \\
=&1-P(\textrm{rank}(A_I)<2k \textrm{ for some } I \textrm { with }
|I|=2k) \nonumber \\
\geq & 1-{n\choose2k}P_d^I  \geq
1-{n\choose2k}e^{-(2k)!(\frac{1}{2})^{4k^2+2k-1}\lfloor\frac{m}{2k}\rfloor},\label{eqn:pind}
\end{align}
where the first inequality follows from the union bound. Then if
$m=g(k) \log n= (2k+1)2^{4k^2+2k-1}\log n /(2k-1)!$,
from (\ref{eqn:pind}) we have the probability that every $2k$ columns of $A$ are linearly
independent is at least $1-1/((2k)!n)$.
Then the statement follows.
\subsection{Proof of Proposition \ref{thm:np}}

Since checking whether or not $r$ given sets 
form an $r$-partition takes polynomial time, $r$-partition problem is NP.

We will show the $r$-partition problem is NP-complete for $r \geq 3$ by proving
that the NP-complete $r$-coloring problem ($r \geq 3$) is polynomial time reducible to the
$r$-partition problem.

Let $G = (V,E)$ and an integer $r$ be an instance of
$r$-coloring. For every $(u, v) \in E$, add a node
$w$ and two links $(w,u)$ and $(w, v)$. Let $W$ denote the set of nodes added. Add a link between every pair of nodes in $V$
not already joined by a link. Let $H$ denote the augmented graph and let $V'$ denote the set of nodes in $H$. 
We claim that if there exists an $r$-partition of $H$, then we can obtain an $r$-coloring of $G$, and vice versa. 

Suppose $S_i$ ($i=1,...,r$) form an $r$-partition of $H$. 
Note that for every $(u, v) \in E$,  $u$ and $v$ cannot belong to the same set $S_i$ for any $i$. 
Suppose $u$ and $v$ both belong to $S_i$ for some $i$. Let $w$ denote the node in $W$ that only directly connects to $u$ and $v$. If $w \in S_i$, then $w$ has both neighbors in the same set with $w$, contradicting the definition of $r$-partition. If $w \notin S_i$, then $H_{V'\backslash S_i}$ is disconnected since $w$ does not connect to any node in $V'\backslash S_i$. It also contradicts the definition of $r$-partition. Thus, for every $(u,v) \in E$, node $u$ and $v$ belong to two  sets $S_i$ and $S_j$ with $i \neq j$. 
Then we obtain an $r$-coloring of $G$.

Let $C_i \subset V$ ($i=1,...,r$) denote an $r$-coloring of $G$. 
We claim that $N_i=C_i$ ($i=1, ..., r-1$), and $N_r=C_r \cup W$ form an $r$-partition of $H$. 
First note for every $u \in V$, at least one of its neighbors is not in the same set as $u$ since $H_V$ is a complete subgraph. 
For every  $w \in W$, $w$ is directly connected to $u$ and $v$ with $(u,v) \in E$. From the definition of $r$-coloring, 
$u$ and $v$ are in different sets $C_i$ and $C_j$ for some $i \neq j$. Therefore, $w$ has at least one neighbor that is not in $N_r$. Second, we will show $H_{V' \backslash N_i}$ is connected for all $i$. $H_{V' \backslash N_r}$ is in fact a complete graph, and thus connected. For every $i<r$, let $S_i:=V \backslash C_i$, then $V' \backslash N_i= S_i \cup W$. $H_{S_i}$ is a complete subgraph, and thus connected. For every $w  \in W$, since its two neighbors cannot be both in $C_i$, then at least one neighbor belongs to $S_i$, 
thus $H_{V' \backslash N_r}=H_{S_i \cup W}$ is connected. $N_i$ ($i=1,..., r$) indeed forms an $r$-partition.
%

\bibliographystyle{IEEEtranS}

\begin{thebibliography}{10}
\providecommand{\url}[1]{#1}
\csname url@samestyle\endcsname
\providecommand{\newblock}{\relax}
\providecommand{\bibinfo}[2]{#2}
\providecommand{\BIBentrySTDinterwordspacing}{\spaceskip=0pt\relax}
\providecommand{\BIBentryALTinterwordstretchfactor}{4}
\providecommand{\BIBentryALTinterwordspacing}{\spaceskip=\fontdimen2\font plus
\BIBentryALTinterwordstretchfactor\fontdimen3\font minus
  \fontdimen4\font\relax}
\providecommand{\BIBforeignlanguage}[2]{{%
\expandafter\ifx\csname l@#1\endcsname\relax
\typeout{** WARNING: IEEEtranS.bst: No hyphenation pattern has been}%
\typeout{** loaded for the language `#1'. Using the pattern for}%
\typeout{** the default language instead.}%
\else
\language=\csname l@#1\endcsname
\fi
#2}}
\providecommand{\BIBdecl}{\relax}
\BIBdecl

\bibitem{AHSC09}
L.~Applebaum, S.~D. Howard, S.~Searle, and R.~Calderbank, ``Chirp sensing
  codes: Deterministic compressed sensing measurements for fast recovery,''
  \emph{Applied and Computational Harmonic Analysis}, vol.~26, no.~2, pp. 283
  -- 290, 2009.

\bibitem{BTH11}
P.~Babarczi, J.~Tapolcai, and P.-H. Ho, ``Adjacent link failure localization
  with monitoring trails in all-optical mesh networks,'' \emph{IEEE/ACM Trans.
  Netw.}, vol.~19, no.~3, pp. 907 --920, 2011.

\bibitem{BA99}
A.~Barab\'asi and R.~Albert, ``Emergence of scaling in random networks,''
  \emph{Science}, vol. 286, no. 5439, pp. 509--512, 1999.

\bibitem{BGIKS08}
R.~Berinde, A.~Gilbert, P.~Indyk, H.~Karloff, and M.~Strauss., ``Combining
  geometry and combinatorics: a unified approach to sparse signal recovery,''
  \emph{arxiv:0804.4666}, 2008.

\bibitem{Blumensath10}
T.~Blumensath, ``Compressed sensing with nonlinear observations,'' Tech. Rep.,
  2010.

\bibitem{Bollobas01}
B.~Bollobas, \emph{Random Graphs}, 2nd~ed.\hskip 1em plus 0.5em minus
  0.4em\relax Cambridge University Press, 2001.

\bibitem{BDPT02}
T.~Bu, N.~Duffield, F.~L. Presti, and D.~Towsley, ``Network tomography on
  general topologies,'' in \emph{Proc ACM SIGMETRICS}, 2002, pp. 21--30.

\bibitem{CaT05}
E.~Cand\`{e}s and T.~Tao, ``Decoding by linear programming,'' \emph{IEEE Trans.
  Inf. Theory}, vol.~51, no.~12, pp. 4203--4215, 2005.

\bibitem{CaT06}
------, ``Near-optimal signal recovery from random projections: Universal
  encoding strategies?'' \emph{IEEE Trans. Inf. Theory}, vol.~52, no.~12, pp.
  5406--5425, 2006.

\bibitem{CBSK07}
Y.~Chen, D.~Bindel, H.~H. Song, and R.~Katz, ``Algebra-based scalable overlay
  network monitoring: Algorithms, evaluation, and applications,''
  \emph{IEEE/ACM Trans. Netw.}, vol.~15, no.~5, pp. 1084 --1097, 2007.

\bibitem{CKMS10}
M.~Cheraghchi, A.~Karbasi, S.~Mohajer, and V.~Saligrama, ``Graph-constrained
  group testing,'' \emph{arXiv:1001.1445}, 2010.

\bibitem{CHNY02}
A.~Coates, A.~Hero~III, R.~Nowak, and B.~Yu, ``Internet tomography,''
  \emph{IEEE Signal Processing Magazine}, vol.~19, no.~3, pp. 47 --65, 2002.

\bibitem{CPR07}
M.~Coates, Y.~Pointurier, and M.~Rabbat, ``Compressed network monitoring for ip
  and all-optical networks,'' in \emph{Proc. ACM SIGCOMM IMC}, 2007, pp.
  241--252.

\bibitem{CM06}
G.~Cormode and S.~Muthukrishnan, ``Combinatorial algorithms for compressed
  sensing,'' ser. Lecture Notes in Computer Science, 2006, vol. 4056, pp.
  280--294.

\bibitem{DeVore07}
R.~DeVore, ``Deterministic constructions of compressed sensing matrices,''
  \emph{Journal of Complexity}, vol.~23, no. 4-6, pp. 918 -- 925, 2007.

\bibitem{Don06}
D.~Donoho, ``Compressed sensing,'' \emph{IEEE Trans. Inf. Theory}, vol.~52,
  no.~4, pp. 1289--1306, 2006.

\bibitem{DoT05}
D.~Donoho and J.~Tanner, ``Sparse nonnegative solution of underdetermined
  linear equations by linear programming,'' in \emph{Proc. Natl. Acad. Sci.
  U.S.A.}, vol. 102, no.~27, 2005, pp. 9446--9451.

\bibitem{Dorfman43}
R.~Dorfman, ``The detection of defective members of large populations,''
  \emph{Ann. Math. Statist.}, vol.~14, pp. 436--440, 1943.

\bibitem{DH00}
D.-Z. Du and F.~K. Hwang, \emph{{Combinatorial Group Testing and Its
  Applications (Applied Mathematics)}}, 2nd~ed.\hskip 1em plus 0.5em minus
  0.4em\relax {World Scientific Publishing Company}, 2000.

\bibitem{Duffield06}
N.~Duffield, ``Network tomography of binary network performance
  characteristics,'' \emph{IEEE Trans. Inf. Theory}, vol.~52, no.~12, pp. 5373
  --5388, 2006.

\bibitem{FR11}
M.~Firooz and S.~Roy, ``Link delay estimation via expander graphs,''
  \emph{arxiv:1106.0941}, 2011.

\bibitem{GR11}
\BIBentryALTinterwordspacing
A.~Gopalan and S.~Ramasubramanian, ``On identifying additive link metrics using
  linearly independent cycles and paths,'' 2011. [Online]. Available:
  \url{http://www2.engr.arizona.edu/~srini/papers/tomography.pdf}
\BIBentrySTDinterwordspacing

\bibitem{HPWYC07}
N.~Harvey, M.~Patrascu, Y.~Wen, S.~Yekhanin, and V.~Chan, ``Non-adaptive fault
  diagnosis for all-optical networks via combinatorial group testing on
  graphs,'' in \emph{Proc. IEEE INFOCOM}, 2007, pp. 697 --705.

\bibitem{HBRN08}
J.~Haupt, W.~Bajwa, M.~Rabbat, and R.~Nowak, ``Compressed sensing for networked
  data,'' \emph{IEEE Signal Processing Magazine}, vol.~25, no.~2, pp. 92 --101,
  2008.

\bibitem{NT06}
H.~X. Nguyen and P.~Thiran, ``Using end-to-end data to infer lossy links in
  sensor networks,'' in \emph{Proc. IEEE INFOCOM}, 2006, pp. 1 --12.

\bibitem{TWHR11}
J.~Tapolcai, B.~Wu, P.-H. Ho, and L.~R\'{o}nyai, ``A novel approach for failure
  localization in all-optical mesh networks,'' \emph{IEEE/ACM Trans. Netw.},
  vol.~19, pp. 275--285, 2011.

\bibitem{WWGZMM11}
A.~Wagner, J.~Wright, A.~Ganesh, Z.~Zhou, H.~Mobahi, and Y.~Ma, ``Towards a
  practical face recognition system: Robust alignment and illumination by
  sparse representation,'' \emph{IEEE Trans. Pattern Analysis and Machine
  Intelligence}, no.~99, pp. 1--14, 2011.

\bibitem{WS98}
D.~Watts and S.~Strogatz, ``Collective dynamics of 'small-world' networks,''
  \emph{Nature}, vol. 393, pp. 440--442, 1998.

\bibitem{WHTJ10}
B.~Wu, P.-H. Ho, J.~Tapolcai, and X.~Jiang, ``A novel framework of fast and
  unambiguous link failure localization via monitoring trails,'' in \emph{Proc.
  IEEE INFOCOM}, 2010, pp. 1 --5.

\bibitem{XH07}
W.~Xu and B.~Hassibi, ``Efficient compressive sensing with deterministic
  guarantees using expander graphs,'' in \emph{Proc. IEEE ITW}, 2007, pp. 414
  --419.

\bibitem{XMT11}
W.~Xu, E.~Mallada, and A.~Tang, ``Compressive sensing over graphs,'' in
  \emph{Proc. IEEE INFOCOM}, 2011.

\bibitem{ZRWQ09}
Y.~Zhang, M.~Roughan, W.~Willinger, and L.~Qiu, ``Spatio-temporal compressive
  sensing and internet traffic matrices,'' in \emph{Proc. ACM SIGCOMM}, 2009,
  pp. 267--278.

\bibitem{ZCB06}
Y.~Zhao, Y.~Chen, and D.~Bindel, ``Towards unbiased end-to-end network
  diagnosis,'' in \emph{Proc. SIGCOMM}, 2006, pp. 219--230.

\end{thebibliography}

\end{document}